\documentclass{IEEEtran}
\usepackage{subfigure}
\usepackage{moreverb}
\usepackage{epsfig}
\usepackage{amsmath,amssymb,amsthm,mathrsfs,amsfonts,dsfont}
\usepackage{adjustbox,lipsum}
\usepackage{algorithm,algorithmic}
\usepackage{amsfonts}
\usepackage{epsfig}
\usepackage{amssymb}
\usepackage{amsmath}
\usepackage{amsthm}
\usepackage{subfigure}
\usepackage{multirow}
\usepackage{rotating}
\usepackage{graphicx}
\usepackage{tabularx}
\usepackage{array}
\usepackage{anyfontsize}
\usepackage{color,soul}
\usepackage{graphicx,dblfloatfix}
\usepackage{epstopdf}
\usepackage{blindtext}
\usepackage{amsmath}
\usepackage{amsthm,amssymb,amsmath,bm}
\usepackage{subfigure}
\usepackage{amsfonts}
\usepackage{epsfig}
\usepackage{amssymb}
\usepackage{amsmath}
\usepackage{cite}
\usepackage{graphicx}
\usepackage{fancyhdr}
\usepackage{subfigure}
\usepackage[subfigure]{tocloft}
\usepackage[font={small}]{caption}
\usepackage{subfigure}
\usepackage{tabularx}
\usepackage{cite}

\newtheorem{theorem}{Theorem}
\newtheorem{Pro}{Proposition}

\begin{document}

\title{ Robust Radio Resource Allocation in \\MISO-SCMA Assisted C-RAN in 5G Networks
\thanks{
	M. Moltafet is with Center for Wireless Communication (CWC),
	 University of Oulu, Oulu, Finland (mohammad.moltafet@oulu.fi).
	Saeedeh.~Parsaeefard
	is with the Iran Telecommunication Research Center, Tehran, Iran (Saeede.parsaeefard@gmail.com ).
Mohammad~R.~Javan is with the Department of Electrical and Robotic Engineering,
	Shahrood  University of Technology, Shahrood, Iran (javanster@gmail.com ).
	N. Mokari is with the Department of ECE,
	Tarbiat Modares University, Tehran, Iran (nader.mokari@modares.ac.ir).
	}
\author{
	Mohammad.~Moltafet,
Saeedeh.~Parsaeefard, \IEEEmembership{Senior Member, IEEE}, Mohammad~R.~Javan,
\IEEEmembership{Member, IEEE} Nader.~Mokari, \IEEEmembership{Member,
IEEE}}}

\maketitle
\begin{abstract}
	In this paper, by considering multiple slices, a downlink transmission of a sparse code multiple access (SCMA) based cloud-radio access network (C-RAN) is investigated. In this regard, by supposing multiple input and single output (MISO) transmission technology, a novel robust radio resource allocation is proposed where considering uncertain channel state information (CSI), the worst case approach is applied. The main goal of the proposed radio resource allocation is to, maximize the system sum rate with maximum available power at radio remote head (RRH), minimum rate requirement of each slice, maximum frounthaul capacity of each RRH, user association, and SCMA constraints. To solve the proposed optimization problem in an efficient manner, an iterative method is deployed where in each iteration, beamforming and joint codebook allocation and user association subproblem are solved separately. By introducing some auxiliary variables, the joint codebook allocation and user association subproblem is transformed into an integer linear programming, and to solve the beamforming optimization problem, minorization-maximization  algorithm (MMA) is applied. Via numerical results, the performance of the proposed system model versus different system parameters and for different channel models are investigated.
\end{abstract}	
	
\section{Introduction}
{Due to the explosive growth of advent high data rate wireless services and number of users, there are various challenges in the design and management of the cellular networks such as demands for higher spectrum efficiency and massive connectivity while keeping the cost in the most acceptable level for both providers and customers. In this regard, both academic and industrial researchers are working to develop the next generation of cellular networks  e.g., fifth generation wireless networks (5G).}
 
{One of the basic concepts and key enablers to design 5G, in the most efficient and cost effective manner, is to resort to the concept of software defined networking and virtualization where 5G can be sliced between different services with divers quality of service (QoS) requirements \cite{8444568}. In this so called networking slicing approach, providing isolation between slices in order to hold the QoS of each services under any users' variations of other users is of essential which calls for highly efficient resource management in this context \cite{Saeedeh2}.}

{On the other side, to meet the spectrum efficiency demands of data hungry services in 5G, new transmission techniques are under-investigated  where non-orthogonal multiple access (NOMA) techniques are one of the promising approaches here \cite{7973146}.} On of the interesting NOMA techniques to address the mentioned challenges is sparse code multiple access (SCMA) which can enhance the system performance compared to the other access techniques \cite{7037423,8061028}. 
By applying the SCMA, each subcarrier can be reused more than one time in the coverage 	
area of radio remote head (RRH). In the SCMA technique at the transmitter side the signal of various users are sent on a subcarrier  with the  SCMA encoder, and at the receiver side, the signal of each user is detected by exploiting the message passing algorithm (MPA) \cite{6666156}. 

 { 
 The SCMA technique is an appropriate access technique for downlink that can increase the overall  
 throughput through user multiplexing and is well-matched to heavily loaded 5G networks  \cite{7037423,7926318,8108679,7880945,7341100,8334819,8108371,8370025,7339668}. Notably, the main practical challenge with using the SCMA in the downlink transmission is the complexity of the multiuser detection on the receiver side. In other words, the detection complexity is too high for the mobile terminals. To tackle this issue, low complexity methods for detection are needed. Therefore, the trend of researches is to achieve methods that decrease the complexity of the SCMA receiver  \cite{8316261,7841967,7881102,7510754}.}
 
Next generation wireless network should be able  to manage the high density networks with massive connectivity and various required QoS. In order to deal with these requirements and manage system efficiently, hierarchical software defined cloud-radio access network (C-RAN) with centralized management can be an appropriate choice. { The C-RAN architecture contains several  radio remote heads (RRHs) that are responsible for  the transmission and reception processes, one baseband unit (BBU) center that  is responsible for all the processing and management tasks, and frounthaul links which connect the RRHs to the  BBU center.}

In practice, due to the various practical reasons, such as estimation errors, feedback quantization, and hardware limitations, uncertain CSI is available at transmitter. To deal with the critical applications, where user outage is interpolated as system failure \cite{8119799}, robust optimization approaches play a pivotal role to provide the necessary reliability for the system, which is the main focus of this paper.
  
\subsection{Related Works}

This work can be positioned in the intersection of three categories of resource allocation problems, namely, 1) Resource allocation in SCMA based system  2) Resource allocation in C-RANs 3) Robust resource allocation. While each of these categories has a specific body of works, in the followings, we present the most related works.

\subsubsection{Resource Allocation in SCMA Based System}
{Notably, very few works have investigated the radio resource allocation in the SCMA based systems. Most of the existing works study the SCMA structure and  link level performance \cite{7037423,7341306,7841967,7341100}. \cite{7037563} evaluates the energy efficiency of SCMA method via MATLAB software. In \cite{7339668}, a power and codebook assignment method is studied where the joint sum rate and system energy efficiency is maximized. Moreover, the tradeoff between rate and energy in a SCMA based system is addressed. In \cite{7565151}, the authors investigate resource allocation in a device to device communication based cellular network with SCMA access technique. They propose a maximization optimization problem with minimum signal to interference plus noise ratio (SINR) requirement for both device-to device and cellular users. In addition, to solve the proposed optimization problem they exploit a Hypergraph based method. \cite{8061028} proposes a joint power and codebook allocation method to maximize the total sum rate in a SCMA based heterogeneous cellular network (HetNet) where the power domain non-orthogonal multiple access (PD-NOMA) is compared with SCMA from the complexity and total sum rate aspects. In addition, to solve the proposed optimization problem, successive convex approximation (SCA) for low complexity (SCALE) and arithmetic geometric mean approximation (AGMA) methods are deployed.}
\subsubsection{ {Resource Allocation in C-RAN}}
In C-RAN networks multiple centralized radio resource allocation  algorithms are proposed based on the received global information of the networks. In \cite{7840188}, the authors, considering orthogonal frequency division multiple access (OFDMA) based multicell system, propose a centralized radio resource allocation algorithm which maximizes the system sum rate. Moreover, they study the coordinated multi point (CoMP) technique in the considered system model.
   The authors of \cite{7406764} investigate an information-centric network architecture for  device-to-device communications, in which with the software defined radio (SDN) controller, a dynamic centralized radio resource allocation is applied.
      In  \cite{7410051},  the authors develop an algorithm to find both the user association and bandwidth allocation in an  SD-based virtualized information-centric architecture network.
      In \cite{7586708}, considering the effective capacity, an optimal power allocation method is proposed in which the effect of the delay-QoS on power allocation and the gain from content caching is evaluated.

\subsubsection{Robust Resource Allocation} 
Due to dynamic and highly variable nature of wireless channels and users' behavior, considering the error in the system information and specially CSI of resource allocation in wireless networks has been drawn a lot of attention e.g., \cite{parsa010} where different approaches are proposed in this context. Since in 5G, the isolation constraints of each slice is of essential, in this paper, we resort to worst case robust optimization theory where based on the error bound of the CSI minimum value of system sum rate is maximized.  
%
\cite{7442902} investigates the worst case robust beamforming in a PD-NOMA based system. In
\cite{8371013}, considering two types of users as elastic and streaming, sum rate maximization under the worst case uncertainty in a NOMA based system is investigated. 
In \cite{7513455}, supposing that only average
CSI is available at the BS, rate maximization under outage
constraints in a NOMA based system is studied.
\cite{8292371} studies sensitivity analysis of a PD-NOMA based virtualized wireless network  to imperfect successive interference cancellation (SIC). 
\cite{8119799} investigates resource allocation problem for uplink PD-NOMA based networks where  the impact on user and system performance due to errors resulting from imperfect SIC is examined, and the chance constrained robust optimization method to cope with this type of error is deployed.

{As can be concluded non of the aforementioned works investigates the robust radio resource allocation in an SCMA assisted C-RAN considering the multiple input single output (MISO).}

\subsection{Contribution}
The main aim of this paper is to propose a robust radio resource allocation method in a MISO-SCMA assisted C-RAN in the presence of multiple slices. In this regard, using the worst case approach, the minimum value of system sum rate with minimum rate requirement of each slice, capacity limitation of each RRH and maximum transmit power constraints is maximized.

The proposed robust optimization problem is a non-convex problem with mixed integer and continues variables, and therefore, to solve it, the common methods for solving the convex optimization problems can not be used directly. Hence, to solve the proposed resource allocation problem, we introduce the worst case optimization problem of our setup. Afterwards, to deal with the nonlinearity feature of joint codebook allocation and user association subproblem we exploit some auxiliary variables and reformulate the proposed  optimization problem. Finally, we deploy  the SCA based alternate search method (ASM) \cite{4752799,5771610}. Based on the proposed solution, at first by adding some auxiliary variables  the main problem is transformed into a new form where the joint codebook allocation and user association subproblem is in an integer linear form. Therefore, in each iteration, to solve the joint codebook allocation and user association subproblem, an integer linear optimization problem is solved, and for beamforming subproblem a convex optimization problem is solved.

 Our contributions are summarized as follows: 
\begin{itemize}
	\item 
	 {We consider a MISO-SCMA assisted C-RAN in the presence of multiple slices 
	where the minimum rate requirement of each slice and maximum frounthaul capacity of each RRH should be satisfied.  The considered system model is a good match for the next generation wireless networks in which various services with different QoS requirements  have to be served. }
	\item
 {We propose a novel robust radio resource allocation problem in which by determining the power allocation and joint codebook allocation and user association  methods, the system sum rate is maximized. By exploiting this radio resource allocation policy, we maximize the system spectral efficiency while we consider the  reliability for the system with respect to the uncertain CSI.}
	\item 
	We define one parameter for both codebook assignment and user association, which decreases the dimension of the  proposed optimization problem and also decreases the complexity of its solution.
	\item 
	In order to solve the proposed robust optimization problem, at first we utilize the worst case approach, then we develop an ASM base solution with low complexity.	
\end{itemize}

The rest of this paper is organized as follows. In Section \ref{systemmodelandproblemformulation}, system model and problem formulation are presented. In Section \ref{centralized solution algorithm}, the solution algorithm is developed. In Section \ref{numerical results} numerical results are investigated. Finally, in Section \ref{conclusion},  conclusions of the paper are presented.
\section{system model and problem formulation}\label{systemmodelandproblemformulation}

\subsection{System Model}\label{system model}

We consider a scenario with multiple  slices in which  the users of each slice spreading over the total coverage area of a C-RAN. 
  We assume that all the transmitters are equipped with multiple antennas while the receivers are simply single antenna systems. We denote the set of slices by $v\in\mathcal{V}=\{1,\cdots,V\}$, and the set of RRHs  by $b\in\mathcal{B}=\{1,\cdots,B\}$ where $b=1$ shows high power RRH and $b\in\{2,\cdots,B\}$ illustrates the low power RRH. We assume that the set of all users in the network is denoted by $\mathcal{K}=\{1,\cdots,K\}$ which is the union of the set of users of all the slices, i.e., $\mathcal{K}=\cup_{v\in\mathcal{V}}\mathcal{K}_v$ where $\mathcal{K}_v=\{1,\dots,K_v\}$. A typical illustration of the considered system model is illustrated in Fig. \ref{Sysmodel}. We assume that the total bandwidth of the network  is divided into $N$ subcarriers whose bandwidth is less than the coherence bandwidth of the wireless channel. We also denote the channel gain from RRH $b$ to user $k$ over the subcarrier $n$ by $\textbf{h}_{b,n,k}\in \mathcal{R}_c^{M_\text{T}\times 1}$ where $\mathcal{R}_c$ is the complex field, and the beam vector assigned by transmitter $b$ to  user $k$ over subcarrier $n$ by  $\textbf{w}_{b,n,k}\in \mathcal{CW}^{M_\text{T}\times 1}$.
We define an indicator variable $\rho_{b,c,k}\in\{0,1\}$ with $\rho_{b,c,k}=1$ if user $k$ is scheduled to receive information from RRH $b$ over codebook $c$, and $\rho_{b,c,k}=0$ if it is not scheduled to receive from transmitter $b$ over codebook $c$. In addition we define $q_{n,c}\in\{0,1\}$ with $q_{n,c}=1$ if subcarrier $ n $ assigned to codebook $ c  $ and otherwise $q_{n,c}=0$.
 Assume that the information symbol  ${s}_{b,n,k}$ is decided to be transmitted to user $k$ from RRH $b$ over subcarrier $n$. It should be noted that   $s_{b,n,k}$ is the  symbol  after SCMA encoding.
Moreover, for the sake of simplicity,   we consider $\boldsymbol{W}=[\textbf{w}_{b,n,k}], \,\, b\in\mathcal{B}, n\in\mathcal{N},k\in\mathcal{K}$.  

\begin{figure}
	\centering
	\includegraphics[width=.5\textwidth]{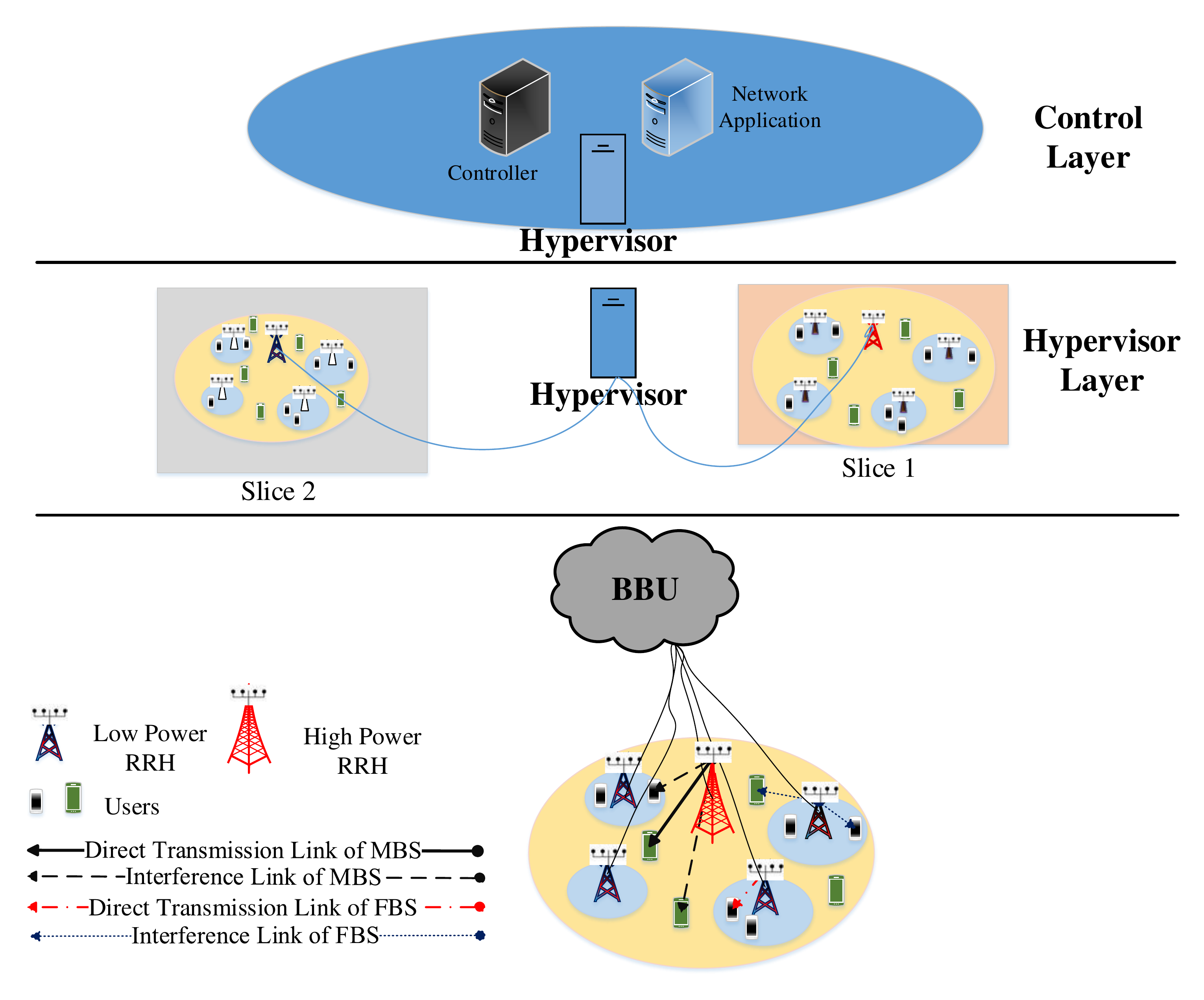}
	\caption{Illustration of the system model containing two slices in a coverage area of one C-RAN with four low power RRHs, one high power RRH, one BBU center, where the fronthauls is the links among RRHs and BBU center.}
	\label{Sysmodel}
\end{figure}
 
In this setup, for the sake of simplicity, we suppose that the frounthaul links (the Links among RRHs and cloud center) are not wireless based link. Therefore, there are not any limitation for frounthaul links. In the SCMA based systems the N-dimensional codewords of a codebook are sparse vectors with $U (U<N)$ non-zero entries.


 {Generally, the detection performance of the SCMA based systems with MPA has the following two properties \cite{748992}: 
 	I) with increasing the complexity of the factor graph, the detection performance is degraded and the complexity of the detection process increases; II) if the factor graph does not have any  loop, the MPA technique provides an exact solution for signal detection. When the number of the multiplexed signals for each subcarrier increases, the factor graph
 	becomes more complex and the probability of existing loops in the factor graph increases. Therefore, the receiver cannot successfully decode the mixed signals.
 	However, under the constraint of the maximum reused degree (i.e., $K_T$) which is presented in \eqref{bnhfg01}, and choosing appropriate $K_T$, different codebooks can be regarded as orthogonal resources approximatively  \cite{7876870,8074731,7848890,7636728,7880945,7339668,8061028,7037563}}

 \begin{align}\label{bnhfg01}
\sum_{b\in\mathcal{B}}\sum_{k\in\mathcal{K}}\sum_{c\in\mathcal{C}}\sum_{n\in\mathcal{N}}  \rho_{b,c,k} q_{n,c}\leq K_\text{T}, \forall n\in\mathcal{N}.
\end{align}
In other words, \eqref{bnhfg01} determines that at most $K_\text{T}$ users could be scheduled on one subcarrier at the same time {due to the above mentioned practical issues}.

{The received signal of user $k$ assigned to RRH $b$ on codebook $c$ is} 
\begin{align}\nonumber
\textbf{y}_{b,c,k}=&\sum_{n\in\mathcal{N}}  \rho_{b,c,k}s_{b,n,k}q_{n,c}\textbf{h}^T_{b,n,k}\textbf{w}_{b,n,k}+\\ \nonumber &\sum_{k'\in\mathcal{K}/k}\sum_{n\in\mathcal{N}} \rho_{b,c,k'}s_{b,n,k'} q_{n,c}\textbf{h}^T_{b,n,k}\textbf{w}_{b,n,k'}+\\ \nonumber &
\sum_{b'\in\mathcal{B}/b}\sum_{k'\in\mathcal{K}}\sum_{n\in\mathcal{N}} \rho_{b',c,k'}s_{b',n,k'} q_{n,c}\textbf{h}^T_{b',n,k}\textbf{w}_{b',n,k'}.
\end{align}

The SINR of user $k$ from transmitter $b$ in codebook $c$ is 
\begin{align}\nonumber
\gamma_{b,c,k}=\frac{\sum_{n\in\mathcal{N}}  q_{n,c}\rho_{b,c,k_{v}}\mid\textbf{w}^T_{b,n,k}\textbf{h}_{b,n,k}\mid^2}{I^{\text{inter}}_{b,c,k}+I^{\text{intra}}_{b,c,k}+ \sigma},
\end{align}
where $\sigma$ is the noise power,
$$I^{\text{intra}}_{b,c,k}=\sum_{k'\in\mathcal{K}/k}\sum_{n\in\mathcal{N}} \rho_{b,c,k'} q_{n,c}\mid\textbf{w}^T_{b,n,k'}\textbf{h}_{b,n,k}\mid^2,$$
and,
\begin{align}
&I^{\text{inter}}_{b,c,k}=\\&\nonumber\sum_{b'\in\mathcal{B}/b}\sum_{k'\in\mathcal{K}}\sum_{n\in\mathcal{N}} \rho_{b',c,k'} q_{n,c}\mid\textbf{w}^T_{b',n,k'}\textbf{h}_{b',n,k}\mid^2.
\end{align}
%

Consequentially, the rate function is given by
\begin{align} 
r_{b,c,k}=\log(1+\gamma_{b,c,k}).
\end{align}

It should be noted that the Shannon's formula gives an upper bound for the rate function in a SCMA based systems \cite{5910654}.

\subsection{Nominal Resource Allocation Problem}
The optimization problem formulation to find the best beamforming, codebook, and user association   approach is formulated as follows:
\begin{subequations}\label{1.5.a}
\begin{align}\label{eeq8a}
&\max_{\mathbf{W},\boldsymbol{\rho}}\; \sum_{b\in\mathcal{B}}\sum_{c\in\mathcal{C}}\sum_{k\in\mathcal{K}} r_{b,c,k}\\& \label{1.5.ab}
 \text{s.t.}:\hspace{.25cm}
 \sum_{c\in\mathcal{C}}\sum_{k\in\mathcal{K}}\sum_{n\in\mathcal{N}} q_{n,c}\rho_{b,c,k}\|\textbf{w}_{b,n,k}\|^2\le P^{b}_{\text{max}}, \\&\nonumber\hspace{1cm} \forall b\in\mathcal{B}, \\&
  \label{1.5.b}
 	\hspace{1cm}\sum_{c\in\mathcal{C}}\sum_{k\in\mathcal{K}}\sum_{n\in\mathcal{N}}  r_{b,c,k}\le R_B^b,\,\,\,\forall b\in \mathcal{B},\\&
 	\label{1.5.c}
\hspace{1cm}\sum_{b\in\mathcal{B}}\sum_{c\in\mathcal{C}}\sum_{k\in\mathcal{K}_v} r_{b,c,k}\ge R^v_{\text{min} },\,\,\,\forall v\in \mathcal{V},\\&\label{1.5.d}
\hspace{1cm}\sum_{b\in\mathcal{B}}\sum_{k\in\mathcal{K}}\sum_{c\in\mathcal{C}}  \rho_{b,c,k} q_{n,c}\leq K_\text{T}, \forall n\in\mathcal{N},
\\\label{constrer2}&
 \hspace{1cm}\rho_{b,c,k}+\rho_{b',c',k}\le 1,\forall b \neq b', b,b'\in \mathcal{B}, c,c' \in \mathcal{C}\\&\label{1.5.f}
\hspace{1cm}\rho_{b,c,k}\in
\begin{Bmatrix}
 0 ,
1
\end{Bmatrix},\,\, b\in\mathcal{B}, c\in\mathcal{C},k\in\mathcal{K}, 
\end{align}
\end{subequations}
where $R_B^b$ is the maximum frounthaul capacity of RRH $b$, \eqref{1.5.ab} indicates the total available
power at each RRH, \eqref{1.5.b} shows the frounthaul capacity limitation for each RRH,
  \eqref{1.5.c} demonstrates the
minimum rate requirement for each slice,  \eqref{1.5.d}  shows the SCMA  constraint  and 
 \eqref{constrer2} shows that each user can be assigned to at most one RRH, simultaneously.
 
 \subsection{Worst Case Resource Allocation Problem }
  In order to 
 model the uncertainty of CSI, the additive error model for  channel model is considered which is given by \cite{6129540}
\begin{align}
\textbf{h}_{b,n,k}=\bar{\textbf{h}}_{b,n,k}+\textbf{e}_{b,n,k},
\end{align}
 where $ \bar{\textbf{h}}_{b,n,k} $   indicates the estimated imperfect channel coefficients and  $ \textbf{e}_{b,n,k} $ denotes the error vector.
  
  Here, we assume uncertain  CSI case in which $ \textbf{e}_{b,n,k} $ is a norm bounded vector for analytical convenience. In
  this regard, by supposing the Euclidean ball-shaped uncertainty the  channel uncertainty  sets are defined as follows:
\begin{align}
&\mathcal{H}_{b,n,k}=\{\textbf{h}_{b,n,k}: \textbf{h}_{b,n,k}=\bar{\textbf{h}}_{b,n,k}+\textbf{e}_{b,n,k}, \|\textbf{e}_{b,n,k}\|\le \kappa\}\\&\nonumber\forall  b\in\mathcal{B}, n\in\mathcal{N}, k\in\mathcal{K}.
\end{align} 
 where $\kappa$  is the error bounds on the uncertainty
 region of the channel coefficient $ \bar{\textbf{h}}_{b,n,k} $.
 
 In the worst-case
 approach,  the considered value of channel coefficients  have to belong to  the related  uncertainty sets, and also  they have to provides the minimum rate on each codebook for the assigned  considered user. 
 Here, the worst-case approach is transformed into the protection function which simplifies the robust problem considerably \cite{8371013,7324447}.
 
  Therefore, instead of  $ \mid\textbf{w}^T_{b,n,k}\textbf{h}_{b,n,k}\mid^2 $,
 we can rewrite
 \begin{align}\label{worstchannel}
&\mid\textbf{w}^T_{b,n,k}\textbf{h}_{b,n,k}\mid^2=\\&\nonumber\textbf{w}^T_{b,n,k}(\bar{\textbf{h}}_{b,n,k}+\textbf{e}_{b,n,k})(\bar{\textbf{h}}_{b,n,k}+\textbf{e}_{b,n,k})^T \textbf{w}_{b,n,k}=\\&\nonumber\textbf{w}^T_{b,n,k}(\bar{\textbf{h}}_{b,n,k}\bar{\textbf{h}}^T_{b,n,k}+\bold{\Delta}_{b,n,k})\textbf{w}_{b,n,k}=\\&\nonumber \textbf{w}^T_{b,n,k}\bar{\textbf{h}}_{b,n,k}\bar{\textbf{h}}^T_{b,n,k}\textbf{w}_{b,n,k} +\textbf{w}^T_{b,n,k} \bold{\Delta}_{b,n,k}\textbf{w}_{b,n,k}
=\\&\nonumber \mid\textbf{w}^T_{b,n,k}\bar{\textbf{h}}_{b,n,k}\mid^2 +\textbf{w}^T_{b,n,k} \bold{\Delta}_{b,n,k}\textbf{w}_{b,n,k},
 \end{align}
  where 
\begin{align}
  \bold{\Delta}_{b,n,k}=\bar{\textbf{h}}_{b,n,k}\textbf{e}^T_{b,n,k}+ \textbf{e}_{b,n,k}\bar{\textbf{h}}^T_{b,n,k}+\textbf{e}_{b,n,k}\textbf{e}^T_{b,n,k}.
\end{align}
$\bold{\Delta}_{b,n,k} $ is defined as a norm-bounded matrix, i.e., $\| \bold{\Delta}_{b,n,k}\|\le\ominus_{b,n,k}$ where $ \ominus_{b,n,k} $ is determined by  Proposition \ref{pro1}.
\begin{Pro}\label{pro1}
With considering $\|\textbf{e}_{b,n,k}\|\le \kappa$, bound $\ominus_{b,n,k}$ is equal to $\kappa^2+2\kappa\|\bar{\textbf{h}}_{b,n,k}\|$.
\end{Pro}
\begin{proof}  Based on the triangle inequality, we have \cite{Boydconvex}
 \begin{align}\label{norm-bounded}
 &\|\bold{\Delta}_{b,n,k}\|=\|\bar{\textbf{h}}_{b,n,k}\textbf{e}^T_{b,n,k}+ \textbf{e}_{b,n,k}\bar{\textbf{h}}^T_{b,n,k}+\textbf{e}_{b,n,k}\textbf{e}^T_{b,n,k}\|\le\\&\nonumber\|\bar{\textbf{h}}_{b,n,k}\textbf{e}^T_{b,n,k}\|+\|\textbf{e}_{b,n,k}\bar{\textbf{h}}^T_{b,n,k}\|+\|\textbf{e}_{b,n,k}\textbf{e}^T_{b,n,k}\|=\\&\nonumber
 \kappa^2+2\kappa\|\bar{\textbf{h}}_{b,n,k}\| \Rightarrow \ominus_{b,n,k}= \kappa^2+2\kappa\|\bar{\textbf{h}}_{b,n,k}\|.
 \end{align}
 	\end{proof}
Based on the worst case approach, considering the error bound of the CSI, the rate functions for the objective function and constrain \eqref{1.5.c} should give the minimum achievable rate.
Consequently, from \eqref{worstchannel} and \eqref{norm-bounded}, the worst-case SINR of user $k$ from transmitter $b$ in codebook $c$ is defined as follows 
  \begin{align}\nonumber
  &\hat{\gamma}_{b,c,k}=\\&\nonumber\frac{\sum_{n\in\mathcal{N}}  q_{n,c}\rho_{b,c,k_{v}}\big(\mid\textbf{w}^T_{b,n,k}\bar{\textbf{h}}_{b,n,k}\mid^2 -\ominus_{b,n,k} \|\textbf{w}_{b,n,k}\|^2\big)}{\hat{I}^{\text{inter}}_{b,c,k}+\hat{I}^{\text{intra}}_{b,c,k}+ \sigma},
  \end{align}
where $\sigma$ is the noise power, 
  \begin{align}\nonumber
  &\hat{I}^{\text{intra}}_{b,c,k}=\sum_{k'\in\mathcal{K}/k}\sum_{n\in\mathcal{N}}\\&\nonumber \rho_{b,c,k'} q_{n,c}\big(\mid\textbf{w}^T_{b,n,k'}\bar{\textbf{h}}_{b,n,k}\mid^2 +\ominus_{b,n,k} \|\textbf{w}_{b,n,k'}\|^2\big),
  \end{align}
  and,
  \begin{align}\nonumber
  &\hat{I}^{\text{inter}}_{b,c,k}=\sum_{b'\in\mathcal{B}/b}\sum_{k'\in\mathcal{K}}\sum_{n\in\mathcal{N}}\\&\nonumber \rho_{b',c,k'} q_{n,c}\big(\mid\textbf{w}^T_{b',n,k'}\bar{\textbf{h}}_{b',n,k}\mid^2 +\ominus_{b,n,k} \|\textbf{w}_{b',n,k'}\|^2\big).
  \end{align}
Consequently, the rate function for  user $k$ from transmitter $b$ in codebook $c$ in the  objective function and constraint \eqref{1.5.c} is given by 
  $\hat{r}_{b,c,k}=\log(1+\hat{\gamma}_{b,c,k})$. 
$\hat{r}_{b,c,k}$ gives the minimum achievable rate, since based on the error bound of the CSI, the numerator of the SINR has the  minimum value and the denominator has the maximum value.
  
  Based on the worst case approach, considering the error bound of the CSI, the rate functions for constraint  \eqref{1.5.b} should give the maximum achievable rate. Therefore, from \eqref{worstchannel} and \eqref{norm-bounded}, SINR and rate functions in constraint \eqref{1.5.b}  are given as follows:
   \begin{align}\nonumber
  &\overline{\gamma}_{b,c,k}=\\&\nonumber\frac{\sum_{n\in\mathcal{N}}  q_{n,c}\rho_{b,c,k_{v}}\big(\mid\textbf{w}^T_{b,n,k}\bar{\textbf{h}}_{b,n,k}\mid^2 +\ominus_{b,n,k} \|\textbf{w}_{b,n,k}\|^2\big)}{\overline{I}^{\text{inter}}_{b,c,k}+\overline{I}^{\text{intra}}_{b,c,k}+ \sigma},
  \end{align}
  where 
  \begin{align}\nonumber
  &\overline{I}^{\text{intra}}_{b,c,k}=\\&\nonumber\sum_{k'\in\mathcal{K}/k}\sum_{n\in\mathcal{N}} \rho_{b,c,k'} q_{n,c}\big(\mid\textbf{w}^T_{b,n,k'}\bar{\textbf{h}}_{b,n,k}\mid^2 -\ominus_{b,n,k} \|\textbf{w}_{b,n,k'}\|^2\big),
  \end{align}
and,
  \begin{align}\nonumber
  &\overline{I}^{\text{inter}}_{b,c,k}=\sum_{b'\in\mathcal{B}/b}\sum_{k'\in\mathcal{K}}\sum_{n\in\mathcal{N}}\\&\nonumber \rho_{b',c,k'} q_{n,c}\big(\mid\textbf{w}^T_{b',n,k'}\bar{\textbf{h}}_{b',n,k}\mid^2 -\ominus_{b,n,k} \|\textbf{w}_{b',n,k'}\|^2\big).
  \end{align}
Consequently, the rate function for  user $k$ from transmitter $b$ in codebook $c$ in constraint \eqref{1.5.b} is  
  $$\overline{r}_{b,c,k}=\log(1+\overline{\gamma}_{b,c,k}).$$ 
 Finally, the worst case optimization problem can be represented as  
  \begin{subequations}\label{W1.5.a}
  	\begin{align}\label{Weeq8a}
  	&\max_{\mathbf{W},\boldsymbol{\rho}}\; \sum_{b\in\mathcal{B}}\sum_{c\in\mathcal{C}}\sum_{k\in\mathcal{K}} \hat{r}_{b,c,k}\\& \nonumber
  	\text{s.t.}:\hspace{.25cm}
  	 \eqref{1.5.ab},  \eqref{1.5.d}, \eqref{constrer2},\eqref{1.5.f}, \\&
  	\label{W1.5.c}
  	\hspace{1cm}\sum_{b\in\mathcal{B}}\sum_{c\in\mathcal{C}}\sum_{k\in\mathcal{K}_v} \hat{r}_{b,c,k}\ge R^v_{\text{min} },\,\,\,\forall v\in \mathcal{V},\\&
  	\label{W1.5.b}
  	\hspace{1cm}\sum_{c\in\mathcal{C}}\sum_{k\in\mathcal{K}}\sum_{n\in\mathcal{N}}  \overline{r}_{b,c,k}\le R_B^b,\,\,\,\forall b\in \mathcal{B}.
  	\end{align}
  \end{subequations}

\section{Proposed Two-Step Iterative Algorithm}\label{centralized solution algorithm}
The proposed nominal optimization problem in \eqref{1.5.a} is a special case of the worst case optimization problem \eqref{W1.5.a} $\kappa=0$, e.g., For $\kappa=0$, we have $\ominus_{b,n,k}=0, \forall b, n, k$ then, \eqref{W1.5.a} is transformed into \eqref{1.5.a}. Consequently, the solution algorithm of the worst-case optimization problem \eqref{W1.5.a} is investigated.

{Since \eqref{W1.5.a} is a non-convex optimization problem containing mixed integer and continues variables, belongs to the NP-hard optimization problem category \cite{4453890,6354282,6415394}. Therefore, finding its optimal solution is not trivial and conventional methods for solving the convex optimization problems cannot be used. Therefore, proposing an efficient algorithm with affordable computational complexity is of essential for this case. Consequently, we resort to the ASM method which is a well-known and efficient algorithm to solve this type of problems and converges to a suboptimal solution and is extensively used in the literature, e.g., \cite{6678362,hooke1961direct}.}

 Based on the ASM approach, an optimization problem is decomposed into several subproblems. The number of subproblems is determined based on the subset of variables. These subproblems are iteratively solved until the convergence is achieved.
{Therefore, for our setup, we consider that} in each iteration, the optimization problem \eqref{W1.5.a} is decoupled into two subproblems referred to the codebook allocation and beamforming subproblems. To solve the beamforming subproblem, first we approximate it by a convex optimization problem and then, we use the interior point method (IPM) to solve it. The main steps of the proposed method is shown in Algorithm \ref{algo1}.
 
\begin{algorithm}
 	\caption{\textit{Two-Step Iterative Algorithm} }
 	\label{algo1}
 	STEP0: Initialization: Set $t=0$ and $T_\text{max}$, set $\boldsymbol{\rho}(0)$, $\mathbf{W}(0)$, and $\epsilon$,\\
 	STEP1: Repeat\\
 	STEP2: Set $\boldsymbol{\rho}=\boldsymbol{\rho}(t)$ and approximate the beamforming problem \eqref{1.6.a} by a convex one,\\
 	STEP3: Find  solution of approximated beamforming problem  and
 	transfer it into $\mathbf{W}(t+1)$,\\
 	STEP4: Find $\boldsymbol{\rho}(t+1)$ by solving the joint codebook allocation and user association subproblem
 	with $\mathbf{W}=\mathbf{W}(t+1)$,\\
 	STEP5: If $t=T_\text{max}$ or $\|\mathbf{W}(t)-\mathbf{W}(t-1)\|\le \epsilon,$
 	stop and return $\boldsymbol{\rho}$ and $\mathbf{W}$, \\\,\,\,\,\,else 
 	set $t=t+1$ and go back to STEP 2.\\
 \end{algorithm}
In the followings, we explain the  approximation method and solution algorithms of the beamforming and joint codebook allocation and user association subproblems. 
 
First, to reach more tractable formulation of \eqref{1.5.a}, we deploy the epigraph transformation \cite{Boydconvex}. Therefore, four auxiliary variables $\boldsymbol{\psi 1}=[\psi 1_{b,c,k}] \, b\in\mathcal{B}, c\in\mathcal{C},k\in\mathcal{K} $, $\boldsymbol{\phi 1}=[\phi 1_{b,c,k}] \, b\in\mathcal{B}, c\in\mathcal{C},k\in\mathcal{K} $, $\boldsymbol{\psi 2}=[\psi 2_{b,c,k}] \, b\in\mathcal{B}, c\in\mathcal{C},k\in\mathcal{K} $, $\boldsymbol{\phi 2}=[\phi 2_{b,c,k}] \, b\in\mathcal{B}, c\in\mathcal{C},k\in\mathcal{K} $ are utilized. By exploiting the auxiliary variables, the main optimization problem is transformed into a new form in which the joint codebook allocation and user association subproblem is a linear programming problem with low complexity and beamforming problem is in a good shape to transform into a convex problem by applying the minorization-maximization algorithm (MMA) approach  \cite{1267055}. Consequently, the main optimization problem can be rewritten as follows: 
   \begin{subequations}\label{1.8.a}
 	\begin{align}
 	&\max_{\mathbf{W},\, \boldsymbol{\phi 1},\boldsymbol{\psi 1}, \boldsymbol{\phi 2},\boldsymbol{\psi 2},\boldsymbol{\rho}}\; \sum_{b\in\mathcal{B}}\sum_{c\in\mathcal{C}}\sum_{k\in\mathcal{K}} \log(\psi 1_{b,c,k})\\& \nonumber
 	\text{s.t.}:\hspace{.25cm} \eqref{1.5.ab}, \eqref{1.5.d}-\eqref{1.5.f},\\&\label{1.8.c}
 	\hspace{1cm}\sum_{b\in\mathcal{B}}\sum_{c\in\mathcal{C}}\sum_{k\in\mathcal{K}_v} \log(\psi 1_{b,c,k})\ge R^v_{\text{min} },\,\,\,\forall v\in \mathcal{V},\\&\label{1.8.d}
 	\hspace{1cm}\sum_{n\in\mathcal{N}}  q_{n,c}\rho_{b,c,k}\big(\mid\textbf{w}^T_{b,n,k}\bar{\textbf{h}}_{b,n,k}\mid^2 -\ominus_{b,n,k} \|\textbf{w}_{b,n,k}\|^2\big)\\&\nonumber \hspace{1cm}\ge 	 \psi 1_{b,c,k} \phi 1_{b,c,k}-\phi 1_{b,c,k}\,\, \forall ,
  b\in\mathcal{B}, c\in\mathcal{C},k\in\mathcal{K},\\&\label{1.8.e}
 	\hspace{1cm} \hat{I}^{\text{inter}}_{b,c,k}+\hat{I}^{\text{intra}}_{b,c,k}+ \sigma\le 
 	\phi 1_{b,c,k},\\&\nonumber 	\hspace{1cm}\,\forall b\in\mathcal{B}, c\in\mathcal{C},k\in\mathcal{K},\\&\label{W2.5.b}
 	\hspace{1cm}\sum_{c\in\mathcal{C}}\sum_{k\in\mathcal{K}}\sum_{n\in\mathcal{N}}  \psi 2_{b,c,k}\le R_B^b,\,\,\,\forall b\in \mathcal{B},\\&\label{w2.8.c}
 	\hspace{1cm}\sum_{n\in\mathcal{N}}  q_{n,c}\rho_{b,c,k}\big(\mid\textbf{w}^T_{b,n,k}\bar{\textbf{h}}_{b,n,k}\mid^2 +\ominus_{b,n,k} \|\textbf{w}_{b,n,k}\|^2\big)\\&\nonumber \hspace{1cm}\le 	 \exp{(\psi 1_{b,c,k})} \phi 2_{b,c,k}-\phi 2_{b,c,k},\\&\nonumber \hspace{1cm}\forall ,
 	b\in\mathcal{B}, c\in\mathcal{C},k\in\mathcal{K},\\&\label{w2.8.e}
 	\hspace{1cm} \overline{I}^{\text{inter}}_{b,c,k}+\overline{I}^{\text{intra}}_{b,c,k}+ \sigma\ge 
 	\phi 2_{b,c,k},\\&\nonumber 	\hspace{1cm}\,\forall b\in\mathcal{B}, c\in\mathcal{C},k\in\mathcal{K}.
 	\end{align}
 \end{subequations}
 
 \subsection{Beamforming Subproblem}
 The beamforming problem with fixed  codebook allocation and user association is 
	\begin{align}\label{1.6.a}
	&\max_{\mathbf{W},\, \boldsymbol{\phi 1},\boldsymbol{\psi 1}, \boldsymbol{\phi 2},\boldsymbol{\psi 2}}\; \sum_{b\in\mathcal{B}}\sum_{c\in\mathcal{C}}\sum_{k\in\mathcal{K}} \log(\psi 1_{b,c,k})\\&\nonumber
	\text{s.t.}:\hspace{.25cm} \eqref{1.5.ab}, \eqref{1.8.c}-\eqref{w2.8.e}. 
	\end{align}
From both left side and right side of constraint  \eqref{1.8.d}, right side of constraint \eqref{w2.8.c}, and left side of constraint \eqref{w2.8.e}, optimization problem \eqref{1.6.a} is a non-convex problem.  To solve the non-convexity issue of the beamformng problem the MMA approach is applied which is based on  SCA. Here, we apply the first order Taylor approximation to transform the non-convex constraints  with a convex one.  
The right side of \eqref{1.8.d} can be written as:
\begin{align}\label{lemma01o}
&\psi 1_{b,c,k} \phi 1_{b,c,k}= \\&\nonumber\frac{1}{4} \bigg[ (\psi 1_{b,c,k}+ \phi 1_{b,c,k})^2 - (\psi 1_{b,c,k}-\phi 1_{b,c,k})^2 \bigg].
\end{align}

By applying the MMA algorithm, the function presented in \eqref{lemma01o}  can be approximated by a convex function as follows:
\begin{align}
&\psi 1_{b,c,k} \phi 1_{b,c,k}\simeq	\frac{1}{4} (\psi 1_{b,c,k} +  \phi 1_{b,c,k})^2 -  \frac{1}{4} \bigg[
(\psi 1^t_{b,n,k}\\&\nonumber -  \phi 1^t_{b,c,k})^2 + 2(\psi 1^t_{b,n,k}-  \phi 1^t_{b,c,k})
(\psi 1_{b,c,k}- \psi 1^t_{b,n,k}-\\&\nonumber \phi 1_{b,c,k}+ \phi 1^t_{b,c,k})
\bigg]=G(\psi 1_{b,n,k}, \phi 1_{b,c,k},\psi 1^{t}_{b,n,k}, \phi 1^{t}_{b,c,k}),
\end{align}
where $\phi^t_{b,c,k}$ and $\psi^t_{b,n,k}$ are calculated by the MMA method as
\begin{align}\label{ineq18c1}
&(\psi 1^{t+1}_{b,n,k}, \phi 1^{t+1}_{b,c,k})=\\&\nonumber\max_{(\psi 1_{b,n,k}, \phi 1_{b,c,k})}G(\psi 1_{b,n,k}, \phi 1_{b,c,k},\psi 1^{t}_{b,n,k}, \phi 1^{t}_{b,c,k}).
\end{align}

For the right side of constraint \eqref{w2.8.c}, the same approach deployed for the right side of  \eqref{1.8.d} is exploited. To approximate the left side of \eqref{1.8.d} to a concave function, we rewrite its first term as
$$
\mid\textbf{w}_{b,n,k}^T\hat{\textbf{h}}_{b,n,k}\mid^2=\|\theta_{b,n,k} \|^2= (\theta^{Rl}_{b,n,k})^2+ (\theta^{Im}_{b,n,k})^2,
$$
where $\theta^{Rl}_{b,n,k}=Real(\textbf{w}_{b,n,k}^T\textbf{h}_{b,n,k})$ and $\theta^{Im}_{b,n,k}=Image(\textbf{w}_{b,n,k}^T\textbf{h}_{b,n,k})$.

Then, by applying the fist order Taylor approximation  around $\theta^t_{b,n,k}$, we can replace   $\| \theta_{b,n,k} \|^2$ as
\begin{align}\label{ineq18c}
&\| \theta_{b,n,k} \|^2 \simeq\| \theta^{t}_{b,n,k} \|^2  + 2 (\theta^{t}_{b,n,k})^\text{T} (\theta_{b,n,k} - \theta^{t}_{b,n,k})\\&\nonumber =g(\theta_{b,n,k},\theta^t_{b,n,k}).
\end{align}

Since $\| \theta_{b,n,k} \|^2\ge g(\theta_{b,n,k},\theta^t_{b,n,k})$, the main idea behind the MMA algorithm is to find $\theta^{t}_{b,n,k}$ for the next iteration which maximize $g(\theta_{b,n,k},\theta^t_{b,n,k})$ as follows \cite{7277111}
\begin{align}\label{properties4o}
\theta^{t+1}_{b,n,k} = \max_{\theta_{b,n,k}} g(\theta_{b,n,k},\theta^t_{b,n,k}).
\end{align}

For the left side of \eqref{w2.8.e}, the mentioned steps used for the left side of \eqref{1.8.d} can be accordingly applied. Consequently, by applying the MMA method, an optimization problem with standard convex form is achieved. To solve the optimization problem IPM can be applied. In order to use the interior point method we use the CVX software \cite{11002233}.

%

 \subsection{Joint Codebook Allocation and User Association Subproblem}
With assuming fixed beamforming variables, the joint codebook allocation and user association subproblem is
 \begin{subequations}\label{1.7.a}
 	\begin{align}
 	&\max_{\boldsymbol{\rho}}\; \sum_{b\in\mathcal{B}}\sum_{c\in\mathcal{C}}\sum_{k\in\mathcal{K}} \log(\psi 1_{b,c,k})\\& 
 	\text{s.t.}:\hspace{.25cm} \hspace{.25cm} \eqref{1.5.ab},  \eqref{1.5.d}-\eqref{1.5.f}, \eqref{1.8.c}-\eqref{w2.8.e}. \nonumber
 	\end{align}
 	 \end{subequations}
Note that \eqref{1.7.a} belongs to the integer linear programming. To solve it the available online software such as  MOSEK  can be deployed \cite{11002233}.
 	
\section{Convergence  and Computational Complexity}
In this section, convergence and computational complexity of the proposed solution are investigated. 

\subsection{Validity of the Approximation Method }	
The SCA method with the linear approach is a well known method to approximate the non-convex optimization problem with a convex one \cite{7277111,6678362,7100916,8456624}. The applied approximation method for the small values of channel gains ($\ll1$) 
has acceptable accuracy \cite{8456624}. In the considered system model, due to the fact that the entries of the channel gain matrices have small values the MMA approximation has good accuracy \cite{8456624}.
 	\subsection{Convergence}
{	The convergence of Algorithm 1 is investigated in the following theorem. }
 	\begin{theorem}
 		\label{proposition300}
 		With the iterative approach presented in Algorithm \ref{algo1}, after each iteration,  the objective function increases compared to the previous iteration, and finally converges.
 	\end{theorem}
 	
 	\begin{proof}
 		Consider  $U(\mathbf{W},\boldsymbol{\rho})=\sum_{b\in\mathcal{B}}\sum_{c\in\mathcal{C}}\sum_{k\in\mathcal{K}} r_{b,c,k}$, in order to achieve  convergence of the algorithm we need 
 		\begin{align}\label{prty}
 		&\dots \le U(\mathbf{W}^{t},\boldsymbol{\rho}^{t})\stackrel{a}{\le} U(\mathbf{W}^{t},\boldsymbol{\rho}^{t+1})\stackrel{b}{\le}\\\nonumber& U(\mathbf{W}^{t+1},\boldsymbol{\rho}^{t+1})\stackrel{c}{\le} \dots.
 		\end{align}
 		Inequality (a) in  \eqref{prty} comes from the fact that optimization problem with variables $\boldsymbol{\rho}$ and fixed $\mathbf{W}$ ($\mathbf{W}=\mathbf{W}^{t}$) is a linear program which its  solution is equal or better than that of  $\boldsymbol{\rho}=\boldsymbol{\rho}^t$. Consequently, we have  $U(\mathbf{W}^{t},\boldsymbol{\rho}^{t})\stackrel{a}{\le} U(\mathbf{W}^{t},\boldsymbol{\rho}^{t+1})$. For inequalities (b) because the final beamforming optimization problem is a convex problem  the same argument as used for inequality (a) can be exploited \cite{4752799,5771610}. The approximated beamforming problem is a convex optimization problem in which the optimal value of it at each iteration is achieved. Because at each iteration the parameters of approximation is updated based on the results of the previous iteration, the results of the solution  and the value of the objective function in each iteration is improved or stay unchanged in respect to the previous iteration.	
 	\end{proof}
 	\subsection{Computational Complexity}
 As explained in Section \ref{centralized solution algorithm}, to solve \eqref{1.8.a} the ASM approach is applied. 
Complexity of the ASM approach is a linear function of the total number of iterations needed to convergence and  complexity of each subproblem. 
In other words, if $L$ shows the number of iterations, $C\Gamma_1$ indicates the complexity of beamforming subproblem and $C\Gamma_2$  demonstrates  the complexity of joint codebook allocation and user association subproblem,  complexity of the main problem is given by 
$$L(C\Gamma_1+C\Gamma_2).$$
{ In each iteration, in order to solve both beamforming and joint codebook allocation and user association subproblems, CVX toolbox is applied. The CVX toolbox exploits IPM to solve an optimization problem \cite{11002233}. Therefore, computational complexity of  each subproblem is given by \cite{7100916,11002233}}

 \begin{align}
 CO_i=\dfrac{\log(\dfrac{NoC_i}{\nu\underline{\epsilon}})}{\log(\xi)},
 \end{align}
{where $i\in\{1,2\}$ with $i=1$ indicating the beamforming subproblem and $i=2$ indicating the joint codebook allocation and user association subproblem, $NoC_i$ is the number of subproblem $i$ constraints, $ \nu $ is the initial	point to approximate the accuracy of the IPM, $ \underline{\epsilon}$ is the stopping criterion of IPM, and $ \xi $ is used to update the accuracy of the IPM. $NoC_i$ for each subproblem is determined as follows:} 
 $$NoC_1=4\times B \times C\times K+2\times B+V,$$
 $$NoC_2=5\times B \times C\times K+C^2\times B^2+N+2\times B+V.$$

%
%
%
%

\section{numerical results}\label{numerical results}
In this section, numerical results for various system parameters in a downlink SCMA based CRAN system are presented to evaluate the performance of the proposed resource allocation approach. The system parameters are considered as follows. There exist one high power RRH with $500$ m radius and $3$ low power RRHs with $20$ m radius.
The total bandwidth is $10$ MHz. The power spectral density  of the received AWGN noise is also set to $-174$ dBm/Hz. Moreover, the channel gain between each user and RRH has Rayle fading with pathloss exponent $3$. The other parameters are variable  described  in the legend or explanation of each figure. All the considered parameters are summarized in Table \ref{table-5}.

\subsection{ {Comparison of SCMA and Traditional Approaches}}
In order to investigate the performance of a SCMA based C-RAN compared to traditional approaches, we consider an OFDMA based C-RAN as a benchmark. All parameters of these two setups are similar excepts that the number of multiplexed signals over each subcarrier. 
Fig. \ref{Compa} shows system sum rate versus the maximum available power for SCMA and OFDMA based systems with $K=10$. This figure presents a comparison between SCMA and OFDMA for both perfect and uncertain CSI cases. From Fig. \ref{Compa}, with the SCMA technology, the system sum rate is increased up to approximately $65\%$ compared to that of OFDMA. This is because via SCMA, each subcarrier can be used in the coverage area of one RRH more than one time without imposing any interference. Due to the uncertain CSI, the system performance for both SCMA and OFDMA decreases. 

 Fig. \ref{conver} depicts convergence of the proposed solution for SCMA and OFDMA based systems. As can be seen, the proposed algorithm converges after a few iterations for both  SCMA and OFDMA based systems. However, for the SCMA based system due to the more constraints and more complicated formulation it takes more iterations.

\subsection{ {Study of the System Parameters}}
Fig. \ref{MVNO} represents sum rate versus the minimum rate requirement of each slice for both perfect and uncertain CSI cases, $V=3$, and $K_v=4,\forall\,\,\,v\in\{1,2,3\}$. Fig. \ref{MVNO} highlights that with increasing the minimum rate requirement of slices, the sum rate decreases. This is because, by increasing the minimum rate requirement, more resources are consumed to satisfy the minimum rate requirements for each slice and there remains less resource leading to smaller amount of sum rate. 

Fig. \ref{MVNO0} demonstrates sum rate versus the different number of slices with a fixed number of users ($K=10$), and minimum rate requirement $0.5$ (bps/Hz).  From this figure, we can see that with increasing the number of slices, the sum rate decreases. This is because, with increasing the number of slices, due to minimum required rate, the feasibility region of optimization problem shrinks. Consequently, the sum rate decreases.
%
 %
 
 Fig. \ref{User} shows sum rate versus the total number of users. As it is seen, the sum rate increases as the number of users increases, as we expected due tp multiuser diversity gain \cite{Tse}.

\subsection{Study of the Users' Density Distribution and Channel Models}
  In Fig. \ref{edge4}, we study the effects of user distribution on the system performance. Also, we compare the  performance of the proposed solution method to the base line method.
 For the base line approach we suppose that users are associated to BSs based on their distance to BSs. In other words, each user are associated to the nearest BS. As can be seen, for cell-edge users user association method has critical role to improve the system performance and coverage. This is because inter cell interference which affects the users throughput, specially for the cell-edge users, can be controlled with the proposed user association algorithm. While for the baseline method, user association is per-determined. Consequently, by exploits the proposed user association method,  coverage of the considered system is improved. 
 
 In Fig. \ref{channel}, we aim to study the effects of different channel models on the performance of a SCMA based system with the MISO technology considering perfect CSI.  This figure depicts sum rate of SCMA for different channel models namely, Rician and Rayleigh. In this figure, the maximum available power is $25$ Watts and $K=10$ with perfect CSI. As can be seen the Rician channel model gives more sum rate than that of the Rayleigh fading model. This is due to that, with the Rician model a line of sight signal path is always exist during the data transmission. Moreover, exploiting multiple antenna at receiver in Rayleigh model gives more gain than that of the Rician model 
 as expected from \cite{Anderegold}.

\begin{table}
	\centering
	\caption{Considered parameters in numerical results}
	\label{table-5}
	\begin{tabular}{ |l|c|}
		\hline
		
		Parameter  & Value of each parameter  \\

		\hline
		
		$M_T$  &
		$1$,	$3$ and $6$\\
		\hline
		
		$K_T$   &
		$6$\\
		\hline
		
		$U$   &
		$2$\\
		\hline
		
		 Number of subcarriers  &
		$8$, $16$\\
		\hline
%
		
     	Pathloss exponent  &
		$-3$\\
		\hline

        Power spectral density  of the received &\\ AWGN noise  &
         $-174$ dBm/Hz\\
         \hline
         
          Maximum transmit power of high power &\\RRH  &
         $40$ Watts\\
         \hline
         
         Maximum transmit power of each low &\\power RRH  &
         $3$ Watts \\
         \hline
          
         Number of low power RRH  &
         $3$\\
         \hline

        High power RRH radius  &
         $500$ m\\
         \hline
            
        Low power RRH radius &
         $20$ m\\
         \hline
         $\kappa$ &
         $0$, $0.05$, $0.1$\\
         \hline
         $R_B^1$ &
         $20$ bps/Hz\\
         \hline
         $R_B^b,\,\,\,\forall b\in \mathcal{B}/{1}$ &
         $5$ bps/Hz\\
         \hline
	\end{tabular}
\end{table}

\begin{figure}
	\centering
	\includegraphics[width=.53\textwidth]{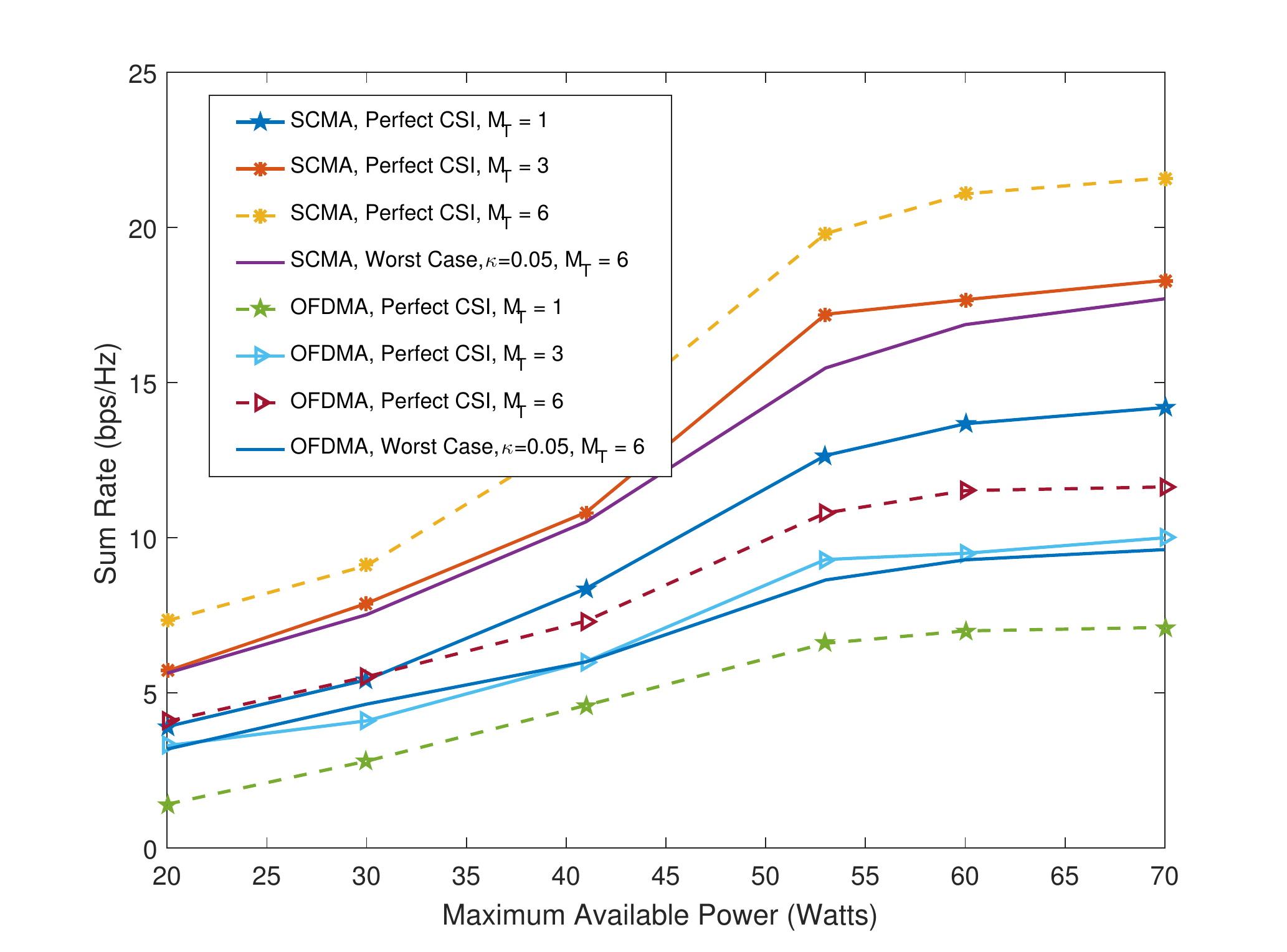}
	\caption{Sum rate versus the maximum available power for SCMA and OFDMA based system.}
	\label{Compa}
\end{figure}

\begin{figure}
	\centering
	\includegraphics[width=.52\textwidth]{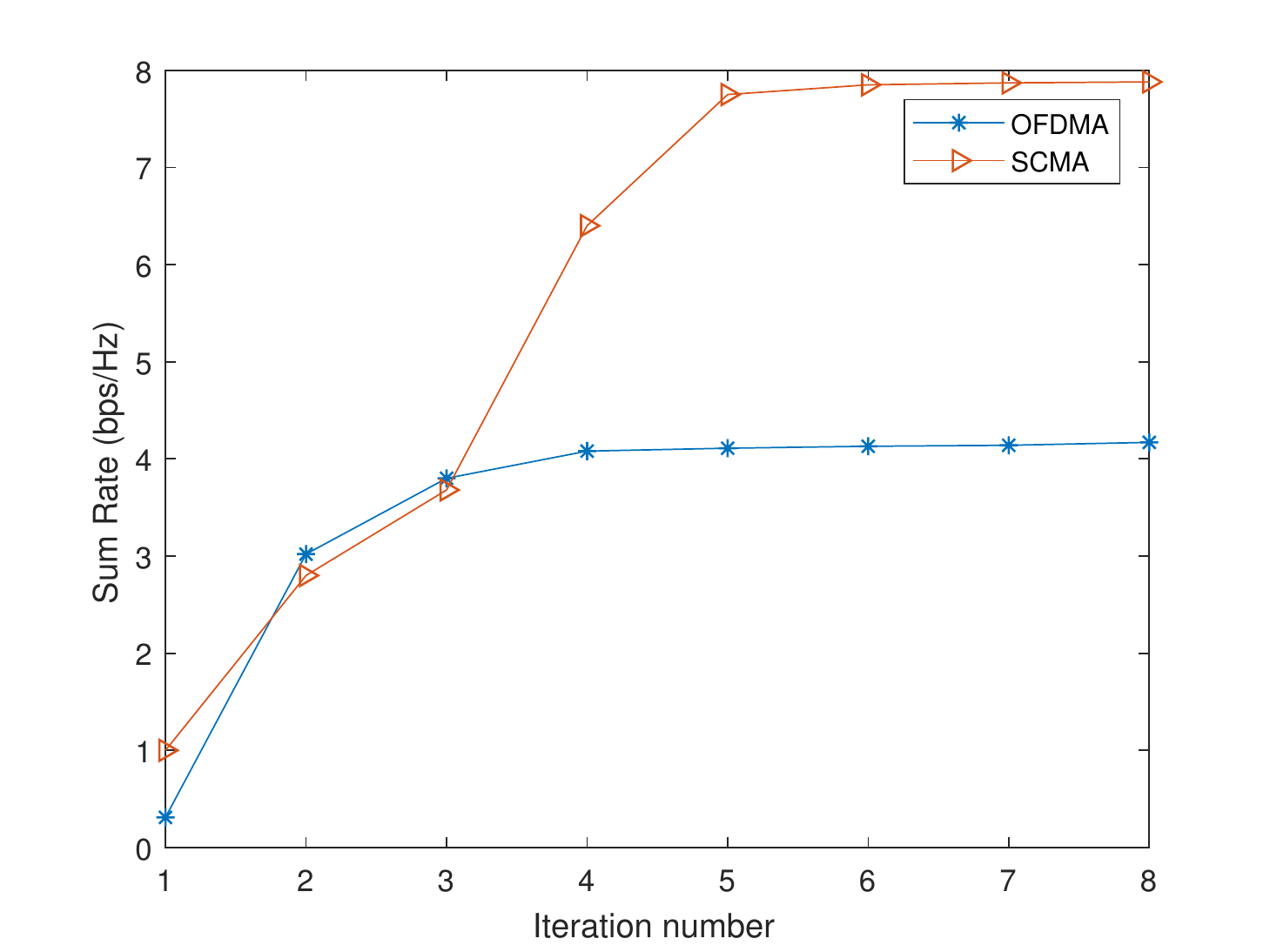}
	\caption{Convergence of the proposed solution for SCMA and OFDMA based systems. }
	\label{conver}
\end{figure}

\begin{figure}
	\centering
	\includegraphics[width=.55\textwidth]{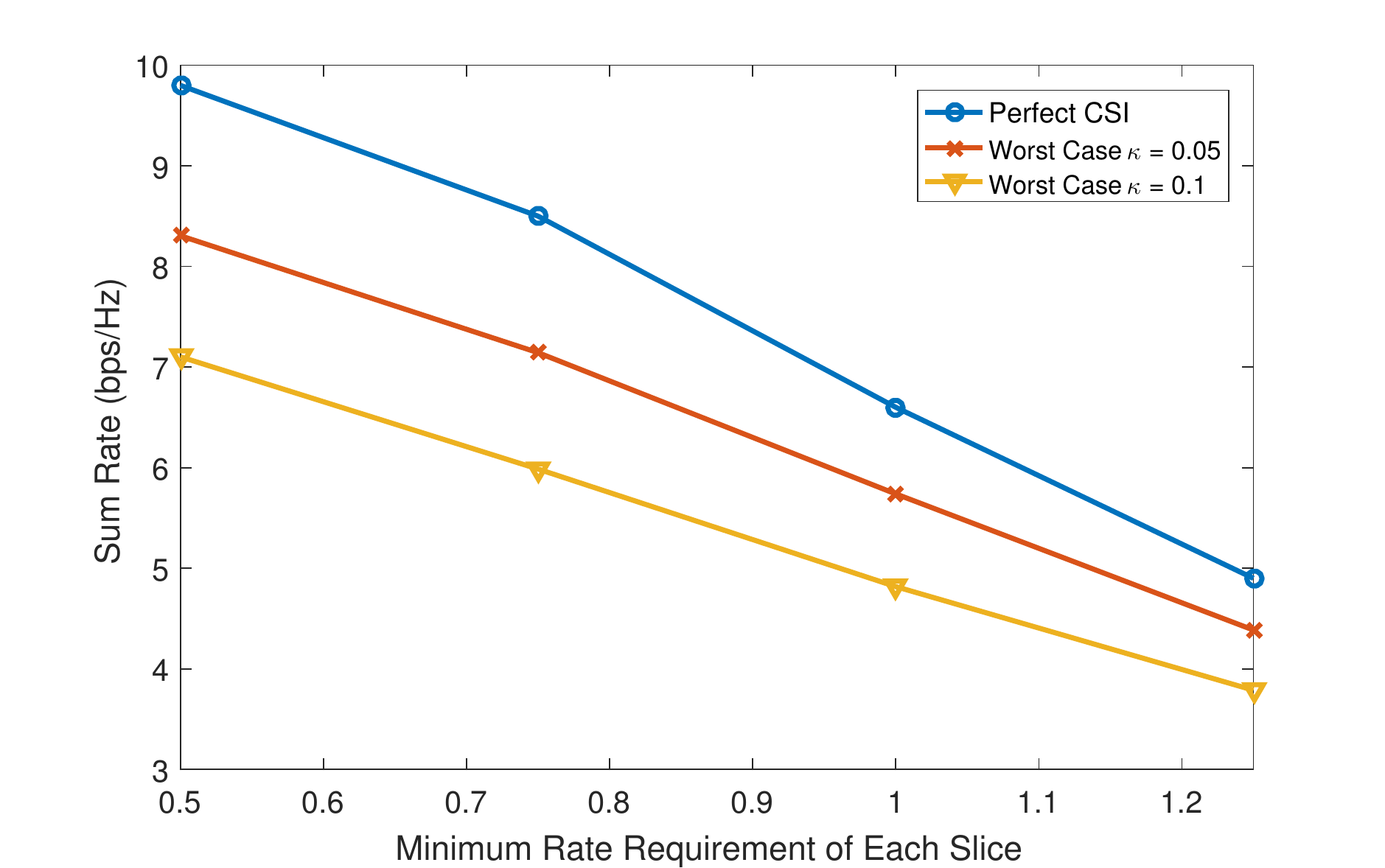}
	\caption{Sum rate versus the minimum rate requirement of each slice.}
	\label{MVNO}
\end{figure}

\begin{figure}
	\centering
	\includegraphics[width=.53\textwidth]{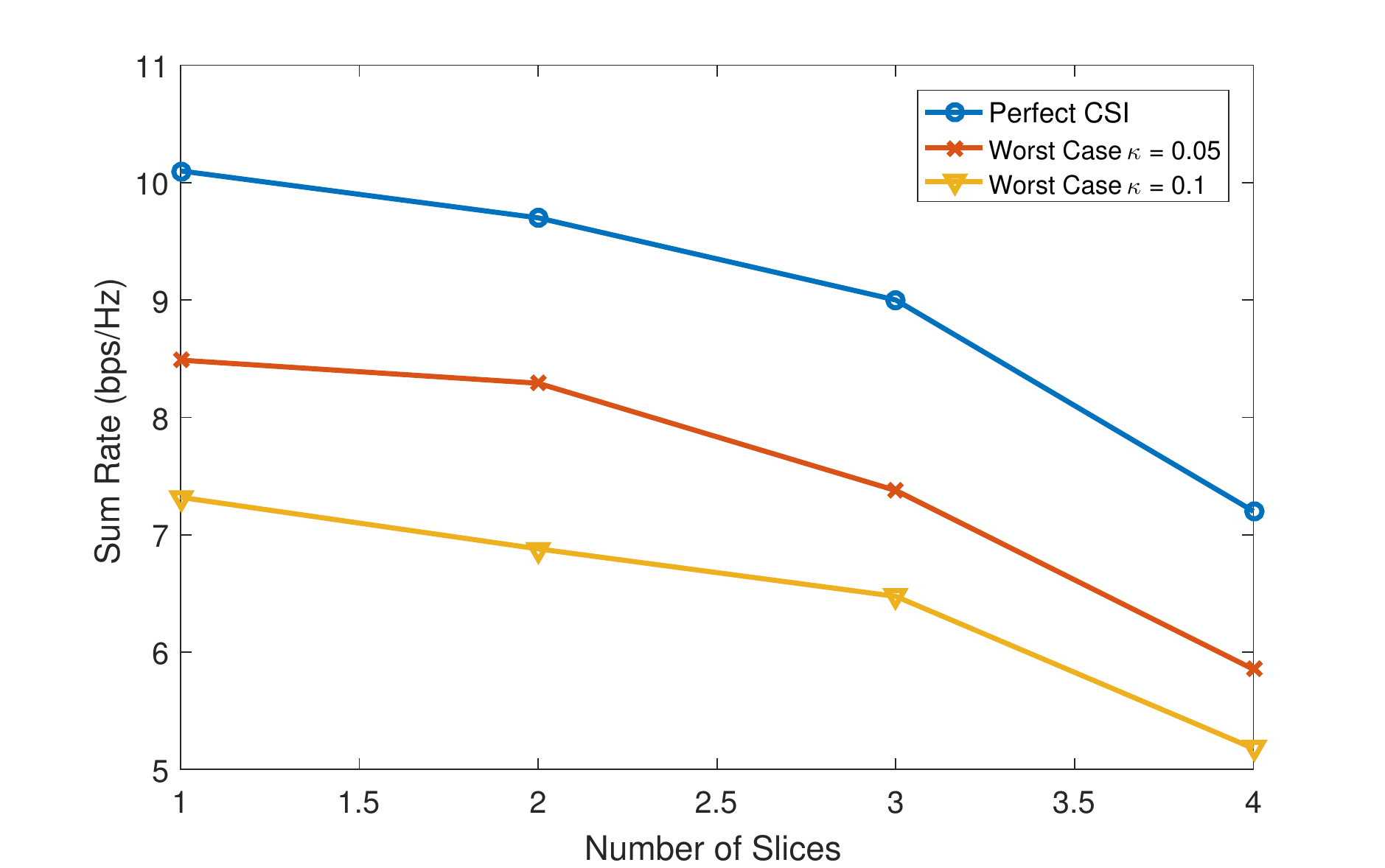}
	\caption{Sum rate versus different number of slices for a fixed number of users.}
	\label{MVNO0}
\end{figure}

\begin{figure}
	\centering
	\includegraphics[width=.52\textwidth]{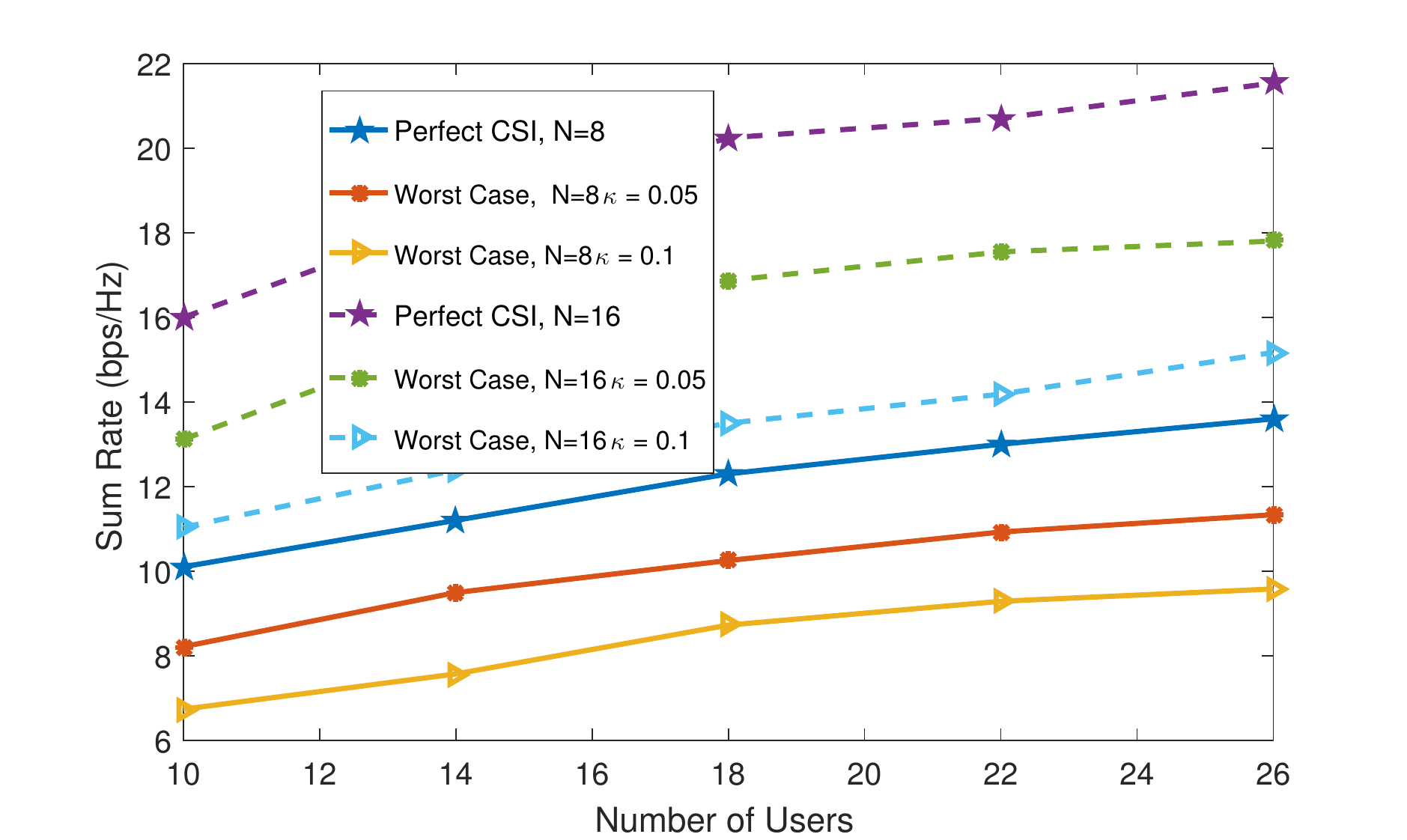}
	\caption{Sum rate versus the total number of users.}
	\label{User}
\end{figure}

\begin{figure}
	\centering
	\includegraphics[width=.54\textwidth]{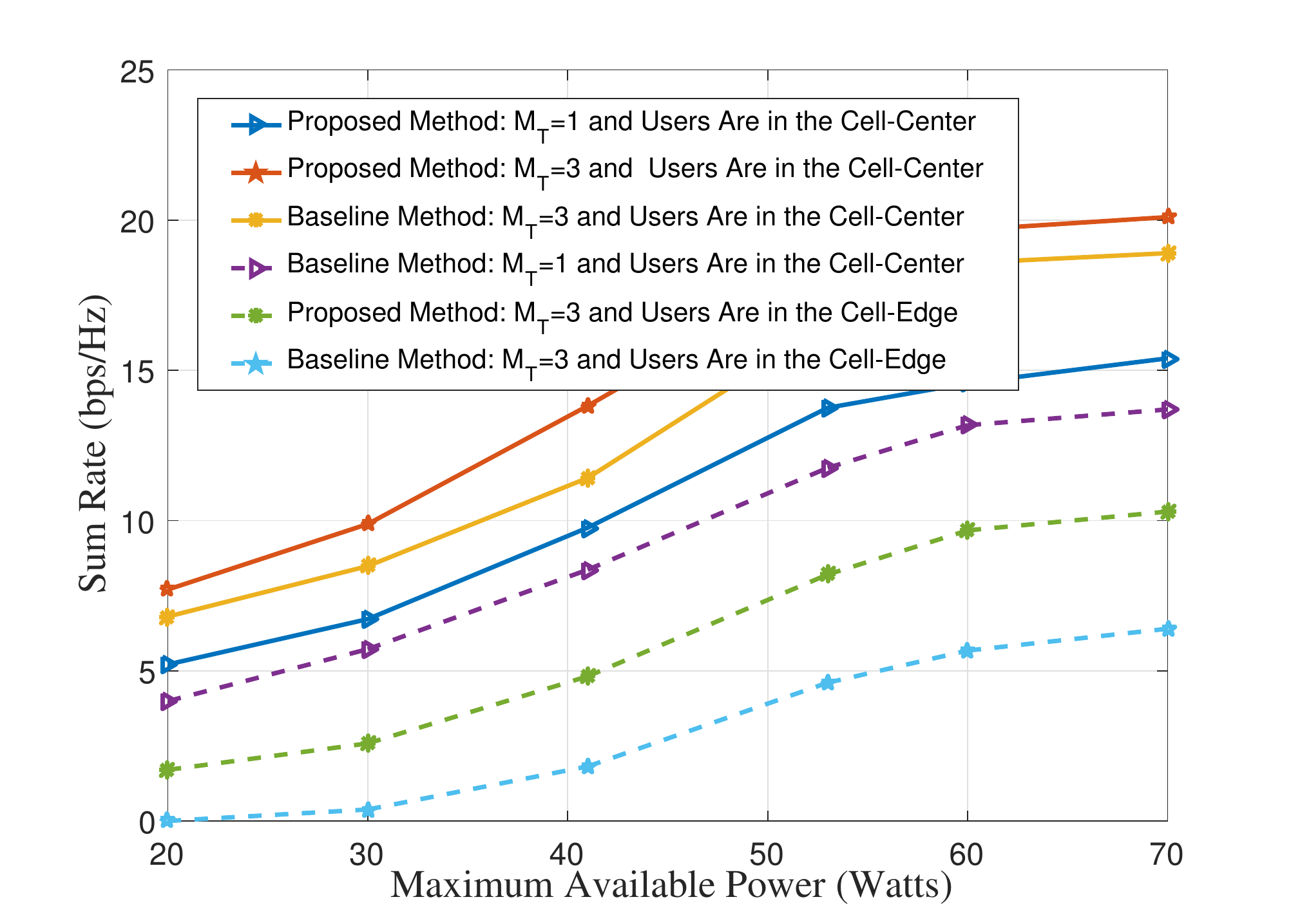}
	\caption{Sum rate versus different user association methods with perfect CSI.}
	\label{edge4}
\end{figure}

\begin{figure}
	\centering
	\includegraphics[width=.47\textwidth]{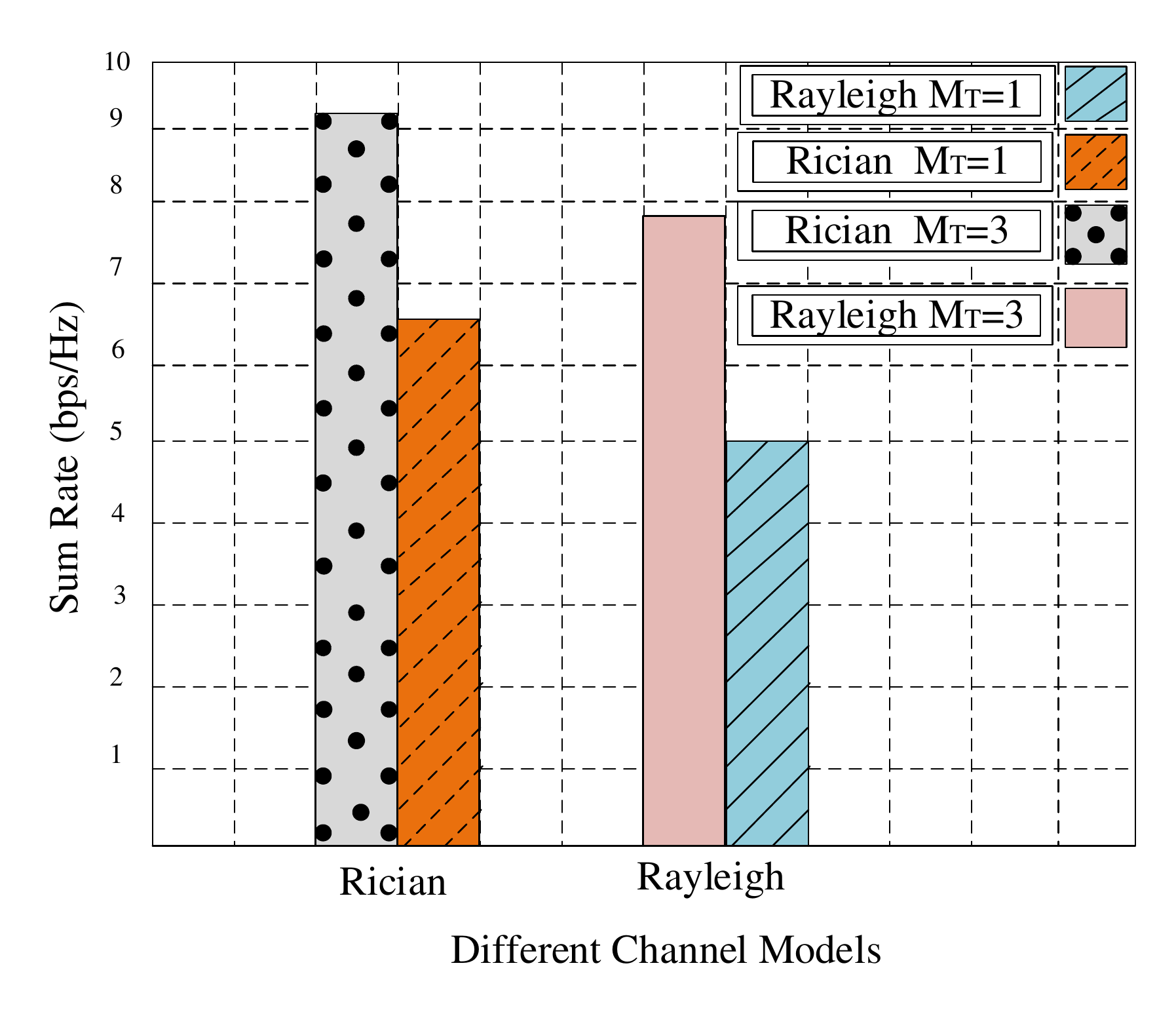}
	\caption{{Sum rate versus the different channel models.}}
	\label{channel}
\end{figure}

 {From the above results, we can conclude that there is a trade off between performance and robustness which depends on the value of $ \kappa $. Consequently, this parameter has an essential role in the system performance.  $ \kappa $ can be determined based on the level of noise in the channel estimation process and probability distribution function of error \cite{6138858,Saeedeh,5164911}.}

 \section{conclusion}\label{conclusion}
 In this paper, considering uncertain CSI, we proposed a worst-case radio resource allocation in a SCMA based C-RAN with MISO transmission. To solve the proposed optimization problem,  we exploited the worst case approach and rewrote the problem via the concept of the protection function to reach more tractable formulation. Then, we applied the ASM method that in each iteration beamforming   and joint codebook allocation and user association subproblems are solved separately and the algorithm is continued until convergence is achieved. 
 Numerical results reveal that the proposed optimization problem via SCMA and MISO technologies improves performance of the system significantly even for the uncertain CSI. This results confess the potential of SCMA to be a perfect candidate for the access technique in 5G. 
 
 \hyphenation{op-tical net-works semi-conduc-tor}
 \bibliographystyle{IEEEtran}
 \bibliography{IEEEabrv,Bibliography}

\end{document}